\let\accentvec\vec
\documentclass{svjour3}                     

\let\vec\accentvec
\smartqed  
\usepackage{graphicx}
\usepackage{amsmath,amsfonts,amssymb,times,graphicx,color}
\usepackage[caption=false]{subfig}
\usepackage{physics}
\usepackage{ntheorem}

\usepackage{tikz}

\usetikzlibrary{cd} 
\usepackage{mathtools}
\mathtoolsset{showonlyrefs}    
\usepackage{verbatim}
\usepackage{ifthen}
\usepackage{mathptmx}

\theoremstyle{plain}
\theoremheaderfont{\itshape\bfseries}
\theorembodyfont{\normalfont}
\theoremprework{\bigskip\hrule}
\theorempostwork{\hrule\bigskip}
\newtheorem*{example*}{Example}

\journalname{Natural Computing}
\begin{document}

\title{Gauge-invariance in cellular automata}


\author{Pablo Arrighi           \and
    Giuseppe Di Molfetta    \and
    Nathana\"el Eon}

\authorrunning{P. Arrighi, G. Di Molfetta, N. Eon}

\institute{P. Arrighi \and N. Eon 
    (\email{nathanael.eon@lis-lab.fr}) 
    \at Aix-Marseille Université, Université de Toulon, CNRS, LIS, Marseille, 
        France
    \and G. Di Molfetta \at 
    Université Publique, CNRS, LIS, Marseille, France and Quantum Computing 
    Center, Keio University, 3-14-1 Hiyoshi, Yokohama, Kanagawa, 223-8522, 
    Japan
}

\date{6 January 2022}

\maketitle

\begin{abstract}
    Gauge-invariance is a fundamental concept in Physics---known to provide
    mathematical justification for the fundamental forces. In this paper, 
    we provide discrete counterparts to the main gauge theoretical concepts
    directly in terms of Cellular Automata. More precisely, the notions of 
    gauge-invariance and gauge-equivalence in Cellular Automata are formalized.
    A step-by-step gauging procedure to enforce this symmetry upon a given
    Cellular Automaton is developed, and three examples of
    gauge-invariant Cellular Automata are examined.
\end{abstract}


\section{Introduction}
\label{sec:intro}

In Physics, symmetries are essential concepts used to derive the laws which 
model nature. Among them, gauge symmetries are central, since they provide 
the mathematical justification for all four fundamental interactions: the 
weak and strong forces (short range interactions), electromagnetism 
\cite{quigg2013gauge} and to some extent gravity (long range interactions).
In Computer Science, cellular automata (CA) constitute the most established 
model of computation that accounts for euclidean space. Yet its origins lies 
in Physics, where they were first used to model hydrodynamics and multi-body 
dynamics, and are now commonly used to model particles or waves. In this paper,
the key notions of gauge-invariance are defined in the discrete model of 
CA, and a counterpart to the gauging procedure---a 
step-by-step method to enforce this symmetry---is formalized 
within CA.

These methods may lead to natural, Physics-inspired CA. More importantly, 
the fields of numerical analysis, quantum simulation, digital Physics are 
constantly looking for discrete schemes that simulate known Physics 
\cite{georgescu2014quantum,notarnicola2020real}. Quite often, these discrete
schemes seek to retain the symmetries of the simulated Physics; whether in 
order to justify the discrete scheme as legitimate, or in order to do 
the Monte Carlo-counting right \cite{hastings1970monte}. Generally 
speaking, since gauge symmetries are essential in Physics, having a 
discrete counterpart of it may also be \cite{Arrighi2020}. 

Interestingly, this way of enforcing local redundancies also bears some 
resemblances with error-correction, as was pointed out in the context of quantum 
computation in \cite{kitaev2003fault,nayak2008non}, and echoes the fascinating question of 
noise resistance within spatially distributed models of computation 
\cite{harao1975fault,Toom}.

Although we authors come from the field of quantum computation and 
simulation, the formalism we use is totally devoid of 
least action principle, or Lagrangian. The notions here are directly 
formulated in terms of the discrete dynamical system. We believe that this 
provides a uniquely direct route to the root concepts. This discrete 
mathematics framework makes the presentation original, and simpler. But it 
also allows for more rigorous definitions, that in turn allow us to prove 
an equivalence lemma.

This work is based on two previous conference papers 
\cite{arrighi2018gauge,arrighi2019non} by the authors. It also integrates 
a quantum example \cite{Arrighi2020} by one of the authors.
The examples expressed here provide what seems to be the simplest
non-trivial Gauge theories so far and illustrates the key concepts. 
Given that Gauge theories are infamously difficult, we think 
this may be a remarkable pedagogical asset.

The paper is organized as follows. In Sec. \ref{sec:def} we introduce the 
formal definitions of cellular automata---both classical and quantum---
and the notions of gauge transformations and gauge symmetry. 
In Sec. \ref{sec:proc},
a discrete counterpart to the gauging procedure is developed and illustrated
through a simple example. It provides the route one may take in order 
to obtain a gauge-invariant CA, starting from one that does not implement 
the symmetry. Sec. \ref{sec:examples} goes through three examples of 
gauge-invariant CA from the literature given in the same framework: 
a very simple classical version \cite{arrighi2018gauge}, 
a generalization to a larger type of gauge transformations 
\cite{arrighi2019non}, and a quantum CA (QCA) \cite{Arrighi2020}.
In Sec. \ref{sec:theorems}, the notion of equivalence between gauge-invariant
CA is formalized and characterized. In Sec. \ref{sec:conclu} we summarize, 
provide related works and perspectives.

\section{Definitions}
\label{sec:def}

\subsection{Cellular automata}
A cellular automaton (CA) is a dynamical system which operates on a discrete,
uniform space and evolves in discrete time steps through the application---
homogeneously across space---of a local operator. Let us make this formal.

\paragraph{Notations.}
For the reader who does not need any reminder about CA, here is a list of 
notations which will be used in the following:
\begin{itemize}
    \item ${\mathbb{Z}^d}$: underlying structure of space with dimension $d$.
    \item $\Sigma$: alphabet.
    \item ${\mathcal{C}}=\Sigma^{\mathbb{Z}^d}$: set of all configurations.
    \item ${\mathcal{N}}$: neighbourhood.
    \item $c_x$ for $c\in{\mathcal{C}}$ and $x\in{\mathbb{Z}^d}$: shorthand for $c(x)$
    \item $c_{t,x}$ for $c\in{\mathcal{C}}$, $x\in{\mathbb{Z}^d}$ and $t\in\mathbb{N}$:
          shorthand for $\big(F^t(c)\big)_x$
    \item $c_{|I}$ for $c\in{\mathcal{C}}$ and $I\subset{\mathbb{Z}^d}$: shorthand for
          $c:I \longrightarrow \Sigma$ the configuration restricted to 
          a set $I$ of specific positions.
\end{itemize}


\paragraph{Space-time representation.} The discrete, uniform space on which
CA are based, is usually the grid ${\mathbb{Z}^d}$ with $d$ the 
dimension---although more general definitions exist, that replace the grid
by bounded degree graphs, typically Cayley graphs. 
The results will be given for ${\mathbb{Z}^d}$ with
$d$ the dimension, and our examples will only be in one-dimension ($d=1$) 
for simplicity.


\paragraph{Alphabet and classical configuration.} 
The \emph{alphabet} $\Sigma$ is a countable---often finite---set.

\begin{definition}[Classical configuration]\label{def:config}
    A \emph{classical configuration} $c$ over an alphabet $\Sigma$
    is a function that associates a state to each point in ${\mathbb{Z}^d}$:
    \begin{equation}
        c: {\mathbb{Z}^d} \longrightarrow \Sigma.
    \end{equation}
    The set of all configurations will be denoted $\mathcal{C}$.
\end{definition}

A configuration should be seen as the state of the CA at a given time.
We use the shorthand notation $c_x = c(x)$ for $x\in{\mathbb{Z}^d}$ and
$c_{|I}$ for the configuration $c$ restricted to the set $I$---i.e. 
$c:I\longrightarrow \Sigma$---for $I\subset{\mathbb{Z}^d}$. The association
of a position and its state is called a \emph{cell}.

\paragraph{Local rule.} Now that we have a way to describe the system
at a given time---using a configuration---we should define the local
rule. A neighbourhood is a \emph{finite} subset of ${\mathbb{Z}^d}$, which is 
denoted ${\mathcal{N}}$. The local rule takes as input a configuration 
restricted to the neighbourhood ${\mathcal{N}}$ of a cell and outputs the next 
value of the cell. 
\begin{equation}
    f: \Sigma^{{\mathcal{N}}} \longrightarrow \Sigma.
\end{equation}
Applying this local rule at every position simultaneously
defines the evolution of a configuration.

\begin{definition}[Cellular Automaton]\label{def:ca}
    A cellular automaton $F$ with alphabet $\Sigma$, dimension $d$
    and neighbourhood ${\mathcal{N}}$ 
    is a function $F:\mathcal{C}\longrightarrow\mathcal{C}$ to
    another configuration by applying a local rule 
    $f:\Sigma^{\mathcal{N}}\longrightarrow 
    \Sigma$ at every position synchronously:
    \begin{equation}
        F(c)_i = f(c_{|i+{\mathcal{N}}})
    \end{equation}
    where $i\in{\mathbb{Z}^d}$.
\end{definition}

Because the CA defines the configuration at time $t+1$
knowing the configuration at time $t$, we will denote by $c_{t,x}$ the
value of a cell at position $x$ and time $t$.

\subsection{Quantum cellular automata}

The definition of QCA used here is commonly known as partitioned 
QCA (PQCA) \cite{schumacher2004reversible,arrighi2019overview}. 
This choice is motivated by the similarity with the classical version while 
not loosing any generality due to the intrinsically universal nature of 
PQCA \cite{arrighi2012partitioned}.


\paragraph{Hilbert space of quantum configurations} 
The quantum configurations differ from classical configurations because they 
require a distinguished element of $\Sigma$ to be called the \emph{empty 
state} and such that only a finite number of cells are not empty.

\begin{definition}[Finite unbounded configurations]\label{def:quantumconf}
    Consider $\Sigma$ the alphabet, with $0$ 
    a distinguished element of $\Sigma$, called the \emph{empty} state. 
    A \emph{finite unbounded configuration} $c$ over $\Sigma$ is a function 
    $c: \mathbb{Z}^d \longrightarrow \Sigma$, 
    such that the set of the
    $(i_1,\ldots,i_d)\in\mathbb{Z}^d$ for which $c_{i_1\ldots i_d}\neq 0$, 
    is finite. 
    The set of all finite unbounded configurations will be denoted 
    $\mathcal{C}_f$.
\end{definition}

The finite unbounded configurations are taken as basis 
to build the \emph{Hilbert space of configurations} which allows for 
superposition of configurations.

\begin{definition}[Hilbert space of configurations]\label{def:hilbconf}
    The {\em Hilbert space of configurations} is that having orthonormal 
    basis $\{\ket{c}\}_{c\in\mathcal{C}_f}$. It will be denoted 
    $\mathcal{H}$.
\end{definition}

\paragraph{PQCA.}
A PQCA works by partitioning the space into supercells, applying a local
unitary operator $U$ on those supercells, and then doing so again at shifted
positions.

\begin{definition}[PQCA]\label{def:pqca}
    A $d$-dimensional {\em partitioned  QCA} (PQCA) $G$ is induced by a 
    {\em scattering unitary} $U$ taking a hypercube of $2^d$ cells into a 
    hypercube of $2^d$ cells, i.e. acting over 
    $\mathcal{H}_{\Sigma}^{\otimes 2^d}$, and preserving quiescence, 
    i.e. $U\ket{0\ldots 0}=\ket{0\ldots 0}$. 
    Let $J=(\bigotimes_{2\mathbb{Z}^d} U)$ over $\mathcal{H}$ and 
    $\tau=\tau_1\ldots\tau_d$ the diagonal translation. The induced global 
    evolution is $J$ at even steps, and $\tau^\dagger J \tau$ at odd steps.
\end{definition}
The local unitary $U$ in a PQCA can be thought of as a reversible version of
the local function $f$ of a classical CA.

\subsection{Gauge-invariance in CA.}

\paragraph{Gauge transformations.} A global gauge transformation is a function
that maps configurations to configurations through the
application of a position-dependent, local gauge transformation at 
every position.

\begin{definition}[Local gauge transformation group]
    \label{def:localgaugetransf} Let $s\in\mathbb{N}$ be the radius. A group
    $G$ of local gauge transformations with alphabet 
    $\Sigma$, radius $s$ and space dimension $d$ is a subgroup of the 
    bijections over $\Sigma^{(2s+1)^d}$ with further requirement that any two 
    local gauge transformations, applied at distinct positions on words of size 
    $(4s+1)^d$, commute. 
    It is extended to act upon $\mathcal{H}_{\Sigma^{(2s+1)^d}}$) by linearity.
    
    We shall 
    use the abuse of notation $g_x\in G$ for a local gauge transformation 
    over $\mathcal{C}$ (or $\mathcal{H}$ in the quantum case)
    where the local transformation is applied around the cell at position $x$,
    i.e. on the cells at positions $\{x-s,\ldots, x+s\}$, and is the
    identity everywhere else.
\end{definition}

In the following, local gauge transformations will refer to the classical or 
quantum case depending on the context. 


\begin{definition}[Global gauge transformation]
    \label{def:gaugetransf}
    Let $\Sigma$ be the alphabet, $s\in\mathbb{N}$ the radius, $d$ the 
    dimension and $G$ a local gauge transformation group with respect to
    these $\Sigma$, $s$ and $d$. A function $\gamma$ 
    is a \emph{gauge transformation} 
    with respect to the local gauge transformation group $G$ if there exists
    a family $(g_x)_{x\in{\mathbb{Z}^d}}$ of elements of $G$ such that:
    \begin{equation}
        \gamma = \prod_{x\in\mathbb{Z}^d} g_x
    \end{equation}
    and this product is unambiguous, because for any $x,y$ in ${\mathbb{Z}^d}$,
    with $x\neq y$, 
    the following commutation relation holds $[g_x, g_y]=0$.
\end{definition}

\begin{remark} At this point, two remarks need to be made about gauge
    transformations :
    \begin{enumerate}\label{rem:gaugetransf}
        \item An element $\gamma$ of $\Gamma$ can now be thought of as a
              configuration with alphabet $G$, where $g_x$ is the local
              state at position $x$.
        \item From now on, we will call \emph{local gauge transformations} the
              elements of $G$ and \emph{gauge transformations} the elements of
              $\Gamma$.
    \end{enumerate}
\end{remark}

\paragraph{Gauge-invariance.} Invariance under $\Gamma$ for a CA $F$ means that
there is a CA $Z: \Gamma \longrightarrow \Gamma$ such that for any gauge
transformation $\gamma \in \Gamma$ the following equality holds:
$Z(\gamma)\circ F = F\circ \gamma$. It means that gauge transforming before
the evolution or afterwards is equivalent. The reason we introduced $Z$ and
did not allow for every possible transformation after the evolution, is because
we want $F$ to be deterministic, from which follows that the gauge 
transformation to be applied
after the evolution should be deterministically computed from the $\gamma$
applied before.

Since the evolution is local, and the gauge transformation is in itself a
configuration (remark-\ref{rem:gaugetransf}), the function $Z$ is
a CA with alphabet $G$.
This leads us to the formal definition-\ref{def:gaugeinv}.

\begin{definition}[Gauge-invariance in CA]\label{def:gaugeinv}
    Let $F$ be a (possibly quantum) CA with alphabet $\Sigma$ and
    space dimension $d$. Let $G$ be a local gauge transformation group.
    Let $\Gamma$ be the corresponding set of gauge transformations.
    
    $F$ is \emph{gauge-invariant} under $\Gamma$ if there exists a cellular 
    automaton $Z$ with alphabet $G$ such that for all $\gamma\in \Gamma$:
    \begin{equation}\label{eqdef:gaugeinv}
        Z(\gamma) \circ F = F \circ \gamma
    \end{equation}
    We say that $F$ is gauge-invariant with respect to $\Gamma$ and $Z$.
\end{definition}

This can be seen as a special commutation relation between the evolution rule
$F$ and the gauge transformation set $\Gamma$. This idea is illustrated in
Fig. \ref{fig:gaugeinvdef}.

In Physics $Z$ is often taken as the identity, thus making
of gauge-invariance a commutation relation.

\begin{figure}[ht!]
    \centering
    \includegraphics[width=0.3\textwidth]{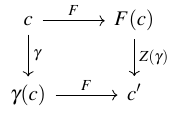}
    \caption{Gauge-invariance}
    \label{fig:gaugeinvdef}
\end{figure}

\section{Gauging procedure.}\label{sec:proc}

Starting from a CA $R$ and a set of gauge transformations
$\Gamma$, it is now possible to check whether $R$ is gauge-invariant under
$\Gamma$ using definition-\ref{def:gaugeinv}. However, in the case that it
is not gauge-invariant, is there a way to extend $R$ into another CA $T$
gauge-invariant under an extension of $\Gamma$? 
In other words, can we always complete $R$ into a gauge-invariant CA? Is 
there a minimal way of doing so?

These questions are very general and, to the best of our knowledge, they do 
not have an answer
in the general case yet. In this section, we provide a guideline in order to
create such a CA $T$. This guideline is called the gauging procedure. Many of
the concepts explored through this procedure, such as the introduction of a
\emph{gauge field}, come from Physics. The procedure itself is Physics-inspired.
It is not a rigorous method that will work in every case, it has
some degrees of freedom and it will need to be adapted depending on the
specific problem---more precisely depending on the structure of the underlying
space-time, the gauge transformation set and the alphabet. Each of the 
following subsections corresponds to a step in the procedure and begins
by developing the general concept before applying it to a running 
example for illustration.

\subsection{Starting point}
The starting point of this procedure is a CA $R$ with state
space $\Sigma$ and a gauge transformation set $\Gamma$ (induced by a
local gauge transformation group $G$). To illustrate the procedure, we will use
a running example in one dimension of space ($d=1$).

\begin{example*}[1]

For our running example we will use a classical, reversible, partitioned CA.
The alphabet for this CA is 
$\Sigma = \{0, 1\}^2$ denoted $\{\square, \blacksquare\}^2$ 
(used in the drawing). We use the following convention to differentiate
the left component from the right component of a cell
$c_x = (c^l_x, c^r_x)$
with $l$ for left and $r$ for right and $c^l_x$ as well as $c^r_x$ in 
$\{0,1\}$.

The local rule
is the one that transports the right component of the state to the right and
the left component of the state to the left. 
Formally, let us denote by $R$ the CA with local evolution $r$.
Focussing on the next left component of the state at position $x$ and 
right
component of the state at position $x+1$ for time $t+1$
we have the following
equation:
\begin{align}\label{eq:r2}
    (c^l_{t+1,x}, c^r_{t+1,x+1}) & = r (c_{t,x}, c_{t,x+1})   \\
                                 & = (c^l_{t,x+1}, c^r_{t,x})
\end{align}

$r$ is not a local rule properly speaking because it does not compute the value
of a specific cell but two components of two different cells. However, it can 
be formulated as a local rule with neighbourhood $1$---i.e. the cell at 
position $x$ is computed from the previous values of the cells at positions 
$x-1$ and $x+1$.

This evolution is illustrated in Fig. \ref{fig:localrule}. 

\begin{figure}[ht!]
    \centering
    \includegraphics[width=0.5\textwidth]{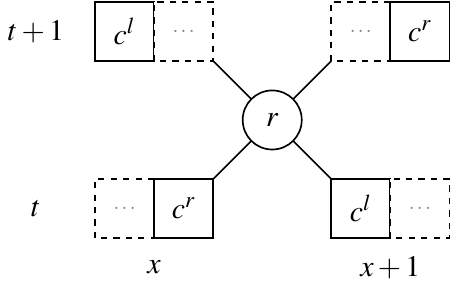}
    \caption{The moving particles of Ex. 1.}
    \label{fig:localrule}
\end{figure}

At this point, one can notice that the CA is made of two independent grids 
(the one for
which $x+t$ is even, and the one for which $x+t$ is odd)
as shown in Fig. \ref{fig:dualgrid}. Indeed, since the evolution takes
the left component of a state to the left and similarly for the right
component, there is no interaction between a cell at position $x$ and one at 
position $x+1$ at time $t$.
Still, both grids are kept here so that the framework 
stays the same for later examples.
\begin{figure}[ht!]
    \centering
    \includegraphics[width=0.7\textwidth]{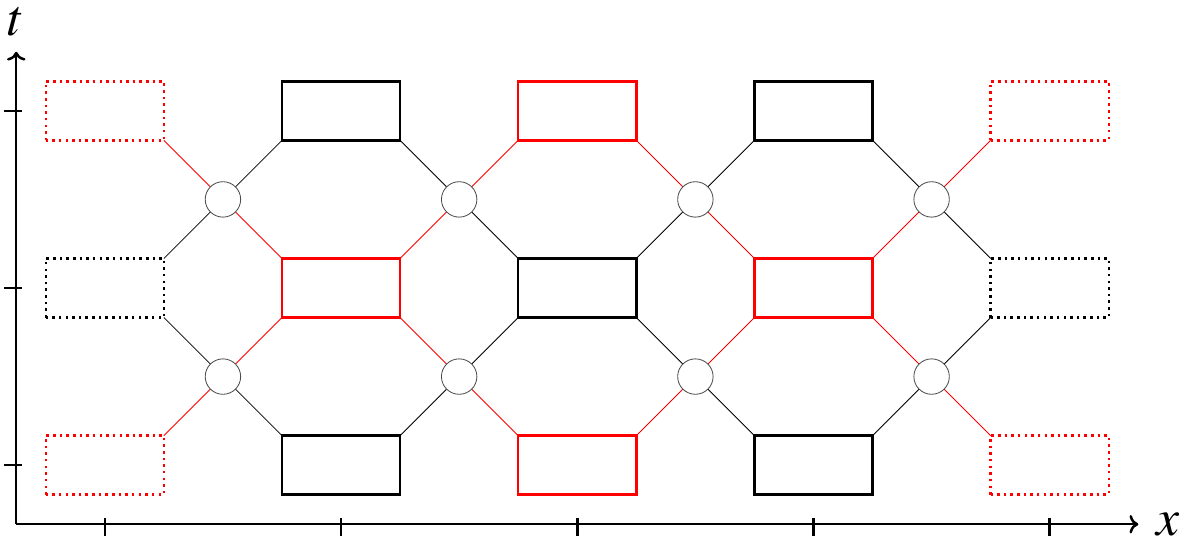}
    \caption{\label{fig:dualgrid}The space-time structure of Ex. 1 is that
        of two independent sub-lattices.}
\end{figure}

In this example, $G= \{ I\otimes I, \tau \otimes\tau\}$ with $I$ the identity, 
$\tau(0) = 1$ and $\tau(1) = 0$.
Changing identically both components
of the state is also what is usually done in Physics.
Notice that $G$ is abelian ; in Sec. \ref{sec:examples}, 
a non-abelian example will be detailed.



\end{example*}

\paragraph{Verification of the gauge-invariance.}
Having a CA $R$ and the set of gauge transformations
$\Gamma$, the next step is to check whether it is gauge-invariant. Showing 
that a CA is not gauge-invariant can
usually be done in quite a straightforward way by applying different gauge 
transformations on the different inputs of the local evolution rule (or 
local unitary in the quantum case).

\begin{example*}[1]
The example is not gauge-invariant. This is illustrated in 
Fig. \ref{fig:nongi} where
the gauge transformation before the evolution (left side of the figure) cannot
be compensated after the evolution (right side of the figure): whatever
local gauge transformation is applied after the evolution, the final state on
both sides of the figure will never match.

\begin{figure}[ht!]
    \centering
    \includegraphics[width=0.55\textwidth]{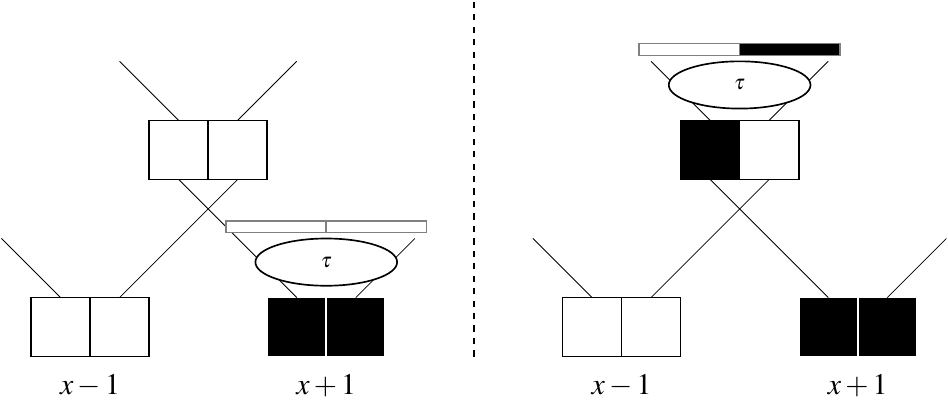}
    \caption{\label{fig:nongi}The initial CA is not gauge-invariant because
        a local transformation done before the evolution cannot be compensated
        for.}
    \label{}
\end{figure}
\end{example*}

\subsection{Introducing the gauge field.}

If the CA is not gauge-invariant, then, following the Physics tradition, we 
seek to extend it into a wider CA, acting over the original field plus a 
\emph{gauge field}.
This gauge field also
changes under a gauge transformation, and it is with respect to this extended
gauge transformation that the extended CA will be gauge-invariant.

\paragraph{Positioning.}
The first question that arises is the
positioning of the gauge field. In the lattice-gauge theory tradition, 
the usual way to position
the gauge field is in-between every two states as shown in
Fig. \ref{fig:framework}.

\begin{figure}[ht!]
    \centering
    \includegraphics[width=.8\textwidth]{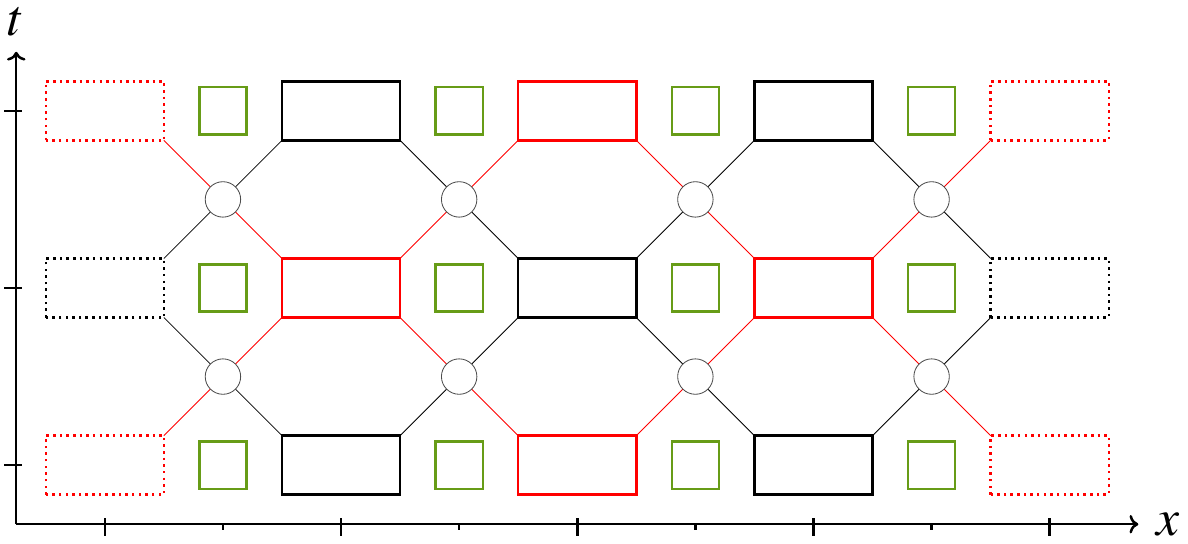}
    \caption{Space-time structure with a gauge field  (in green) in between
        every two cells.}
    \label{fig:framework}
\end{figure}

\begin{example*}[1]
In this example, we make the standard choice of having the gauge field 
positioned in between every two-states. Therefore we will 
denote by $A_{x,x+1}$ the value of the gauge field
in between the position $x$ and $x+1$.
\end{example*}

\paragraph{Set of values.}
The second question that arises is what set of values $\Lambda$ can the gauge 
field take? Usually, one seeks to add the least amount of information that 
will allow for gauge-invariance. Often it turns out that 
the gauge field keeps track of the difference of gauge between the 
states that surrounds it---i.e. it is analogous to a parallel transport 
operator between two separate tangent spaces in a differential manifold. 

\begin{example*}[1]

In this example, we chose $\Lambda=\{I\otimes I,\tau\otimes\tau\}$.
Since there are only two elements in $\Lambda$, the values of this set can be
encoded in a single bit.

This choice is motivated by the fact that the difference of gauge between 
two neighboring positions can only be one of the elements of $\Lambda$.
\end{example*}

In sec-\ref{sec:examples}, other examples for the set of values of the gauge 
field will be developed.

\paragraph{Updating the local rule.} Having defined where and what the gauge
field is, the third question that arises is: how does it interact with a 
configuration? More specifically, how does the local rule (or local unitary)
depend on the gauge field? Usually the gauge field keeps track of the
difference of gauge between neighboring states, and so usually the gauge 
field is used to cancel this difference.

\begin{example*}[1]


One choice for the extended local rule is to harmonize the gauges of the inputs
before applying the previously defined local rule. To do so, in our example 
Eq. \eqref{eq:r2} transforms into:
\begin{align}\label{eq:r3}
    (c^l_{t+1,x}, c^r_{t+1,x+1}) & = r_A (c_{t,x}, c_{t,x+1})                                   \\
                                 & = r\left(A_{x,x+1}\big(c_{t,x}\big), 
                                 A^{-1}_{x,x+1}\big(c_{t,x+1}\big)\right)
\end{align}
where $r_A$ denote the extended, $A$-dependant, local rule.
\end{example*}

\paragraph{Gauge transformation.}
The role of the gauge field at this point is to obtain gauge-invariance.
Writing the gauge-invariance condition \eqref{eqdef:gaugeinv} forces a fourth
question: how does the gauge field transform under a gauge transformation?
To answer this question, one need to write down the gauge-invariance condition,
which puts a constraint over the transformation of the gauge field. 

If the gauge field does not change under a gauge transformation, one can 
easily see that adding the gauge field does not help acquire gauge-invariance, 
since the information added through the gauge field would not help cancel 
the effect of gauge transformations. 
Therefore, the gauge transformations need to be extended to also act upon the 
gauge field. The set of extended gauge transformations will be denoted 
$\overline{G}$.

\begin{example*}[1]
The action of the gauge transformation, in between position $x$ and $x+1$,
over $A_{x,x+1}$ will be denoted by $g(A_{x,x+1})$. This choice will help 
to stay clear of new notations.
The gauge-invariance condition \eqref{eqdef:gaugeinv} puts a constraint which
may wholly determine this action.
Locally, for the running example with $g_x, g_{x+1}$
two elements of $G$, the gauge-invariance condition writes:
\begin{align}\label{eq:localinv}
    r_{g(A)} \circ (g_x \otimes g_{x+1}) = ({Z(g)_x}
    \otimes {Z(g)_{x+1}}) \circ r_A
\end{align}

The gauge-invariance condition requires the existence of a CA $Z$. 
Here we choose $Z=Id$ (but other choices would have been possible).
Thus, Eq. \eqref{eq:localinv} 
transforms into the following two equations:
\begin{equation}
    \begin{cases}\label{eq:localinv2}
        g(A_{x,x+1})^{-1} \circ g_{x} & =
        g_{x+1} \circ A_{x,x+1}           \\
        g(A_{x,x+1}) \circ g_{x+1}    & =
        {g_x} \circ A^{-1}_{x,x+1}
    \end{cases}
\end{equation}

Those two equations are redundant and simplify into:
\begin{equation}
    g(A_{x,x+1}) = g_x \circ A_{x,x+1} \circ g_{x+1}
\end{equation}
Therefore, from the gauge-invariance condition and through the choice of a
$Z$, the extension of gauge transformations to the gauge field is fully
determined. Here they are expressed in a gauge-field centric way, however let 
us express them equivalently in a way that is centered on $x$. 

Indeed, after this extension, the local gauge transformation group $G$ is 
extended into a subgroup $\overline{G}$  of the bijection over 
$\Lambda \otimes \Sigma \otimes \Lambda$. 
Then for any $\overline{g}\in \overline{G}$, 
we define $g \in G$ such 
that for any $A_0,c,A_1\in \Lambda\otimes\Sigma\otimes \Lambda$:
\begin{equation}\label{eq:gaugetransfall}
    \overline{g}\big(A_0,\; c,\; A_1)=\Big(A_0 \circ g,\;
    g(c),\; g \circ A_1\Big)
\end{equation}
The application of this local gauge transformation is 
illustrated in Fig. \ref{fig:gaugetransfext}. 

\begin{figure}[ht!]
    \centering
    \includegraphics[width=\textwidth]{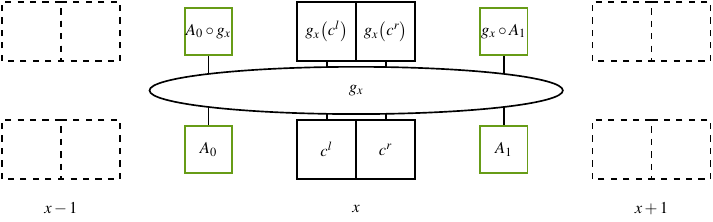}
    \caption{The extended local gauge transformation.}
    \label{fig:gaugetransfext}
\end{figure}

The extended local gauge transformations no longer have disjoint supports, 
but they do commute with one another: 
$$g_{x+1} \circ g_x(A_{x,x+1}) = g_{x+1} \circ 
A \circ g_x = g_x \circ g_{x+1} (A_{x,x+1})$$ for any $A_{x,x+1}$
in $\Lambda$.
Therefore,
the set of extended gauge transformations $\Gamma$, over a full 
configuration and its associated gauge field, can be defined through
definition-\ref{def:gaugetransf}.
\end{example*}

After this step, the CA is gauge-invariant through the use
of an external gauge field and an extended gauge transformation.

\subsection{Dynamics of the gauge field.}
The last step of this procedure is to transform the external gauge field
into an internal state of the CA which will evolve through a local rule. This
leads to the fifth and final question: what is the dynamics of the gauge field?
The choice of dynamics is constrained by the fact that the complete dynamics
should be gauge-invariant---i.e. verify condition \eqref{eqdef:gaugeinv} for 
$\overline{G}$.

\begin{example*}[1]
In our running example a simple gauge-invariant dynamics
for the gauge field, is to choose the identity. Eq. \eqref{eq:r3} thus
transforms into:
\begin{align}\label{eq:rfinal}
    \left(c^l_{t+1,x}, A_{t+1,x,x+1}, c^r_{t+1,x+1}\right) & =
    r\left(c_{t,x}, A_{t,x,x+1}, c_{t,x+1}\right)                                                                              \\
    & = \left({A_{t,x,x+1}}\big(c^l_{t,x+1}\big), \;A_{t,x,x+1}, \;
    {A_{t,x,x+1}} \big(c^r_{t,x}\big)\right)
\end{align}

The identity is indeed gauge-invariant with respect to $\Gamma$ and $Z$ in 
this example. This is due to the fact that $Z$ is the identity and as such,
Eq. \eqref{eqdef:gaugeinv} is a simple commutation relation which is trivially 
true for the identity.

This complete evolution is represented graphically in
Fig. \ref{fig:frameworkcomplete}.
\begin{figure}[ht!]
    \centering
    \includegraphics[width=0.8\textwidth]{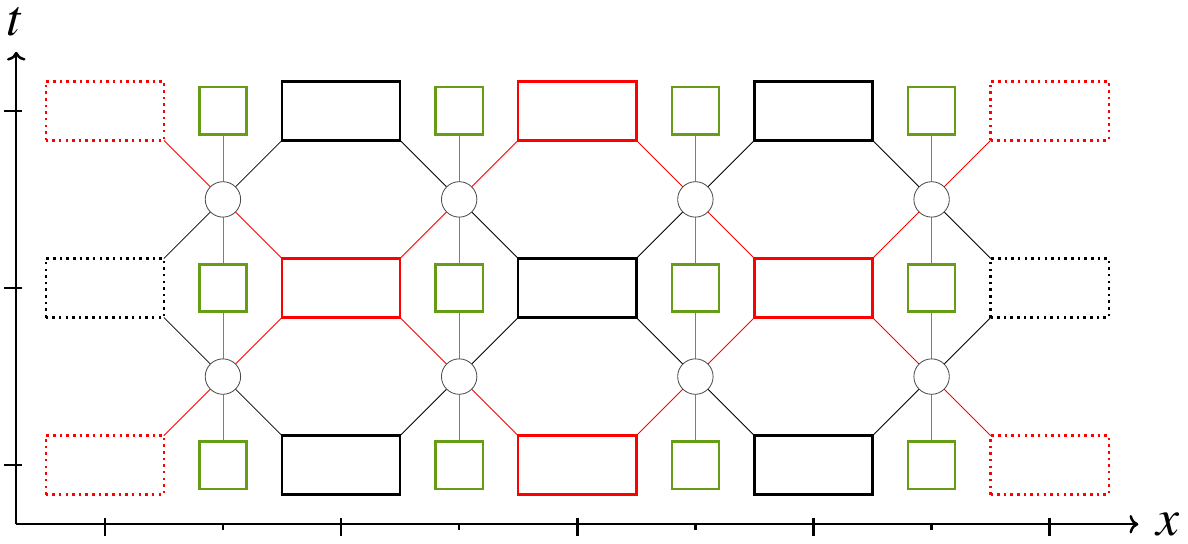}
    \caption{Space-time structure with the gauge field dynamics.}
    \label{fig:frameworkcomplete}
\end{figure}
\end{example*}

\section{Examples of reversible and quantum constructions}
\label{sec:examples}

The procedure helps designing new CA that implement gauge-invariance. In this
section, examples both in classical and quantum settings will be developed.
Every example presented here have the same space-time layout as that of 
Fig. \ref{fig:frameworkcomplete}. This choice is made for 
simplicity and clarity, but there is no claim of generality. Although it is 
known that PQCA are intrinsically universal \cite{arrighi2012partitioned},
it could be that they
are not the most general setting in which to define gauge-invariant CA.

For all of these examples, the local evolution $Z$ of the gauge transformations 
will be taken to be the identity. Then again, this is only a choice.

\begin{example*}[1]
This example has been developed extensively in the previous sections.
Fig. \ref{fig:rcaa} shows three space-time diagrams implementing the 
gauge-invariant rule given in Eq.\eqref{eq:rfinal}. An empty state for the
gauge field (in green) represents the identity while a filled state 
represents $\tau\otimes\tau$.

Sub-figure \ref{sub@subfig:rcaa1} has its gauge field set to the identity,
therefore coincides with $R$, with a "particle" going 
right.
Sub-figure \ref{sub@subfig:rcaa2} features the same physics 
(a particle going right)
but with a gauge transformation initially applied at the central position.
This is therefore understood as an equivalent situation, expressed differently.
Indeed, if a gauge transformation is later applied at the central position, 
the final configuration yields back that of sub-figure \ref{sub@subfig:rcaa1}.
Finally, Sub-figure \ref{sub@subfig:rcaa3} starts with a similar configuration
as sub-figure \ref{sub@subfig:rcaa1} but with one difference in the gauge
field, that does not come from having applied a gauge transformation.
Both diagrams end up very different. This last sub-figure shows 
that new dynamical behaviours arise, 
which could not have been witnessed without a gauge field.

\begin{figure}[ht!]
    \subfloat[With identity as gauge field\label{subfig:rcaa1}]
    {\includegraphics[width=.49\textwidth]{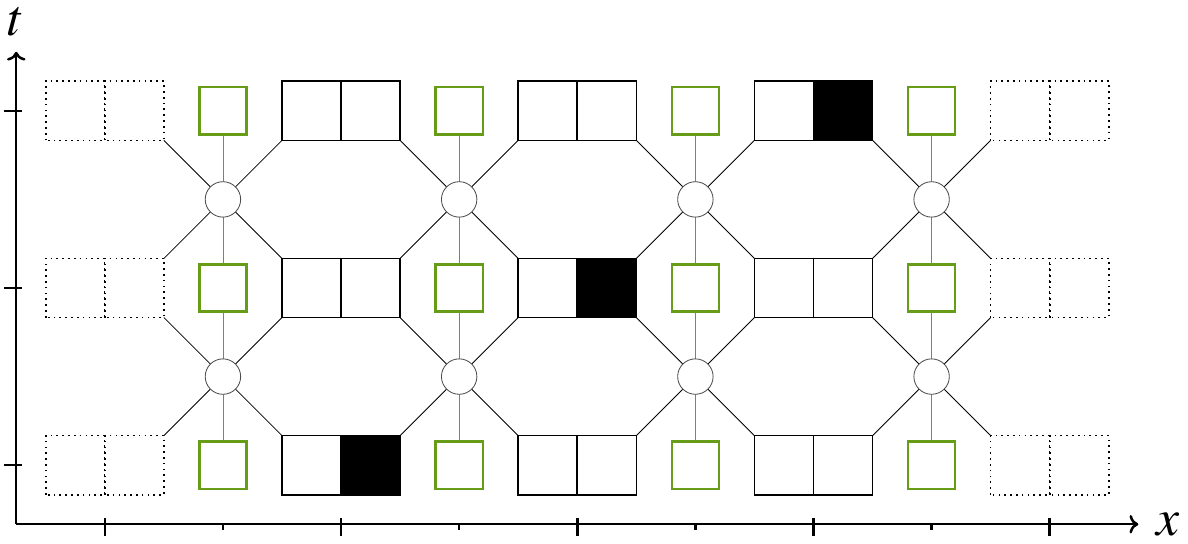}}
    \hfill
    \subfloat[Equivalent dynamics up to a gauge transformation
        \label{subfig:rcaa2}]
    {\includegraphics[width=0.49\textwidth]{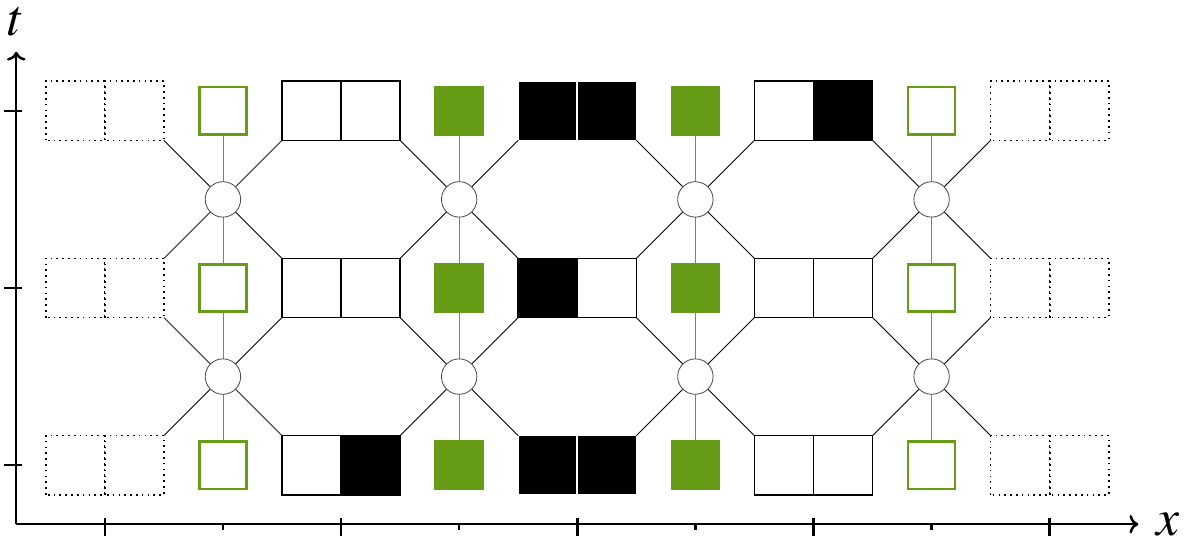}}
    \newline \centering
    \subfloat[New behaviours emerge from having the gauge field
        \label{subfig:rcaa3}]
    {\includegraphics[width=.49\textwidth]{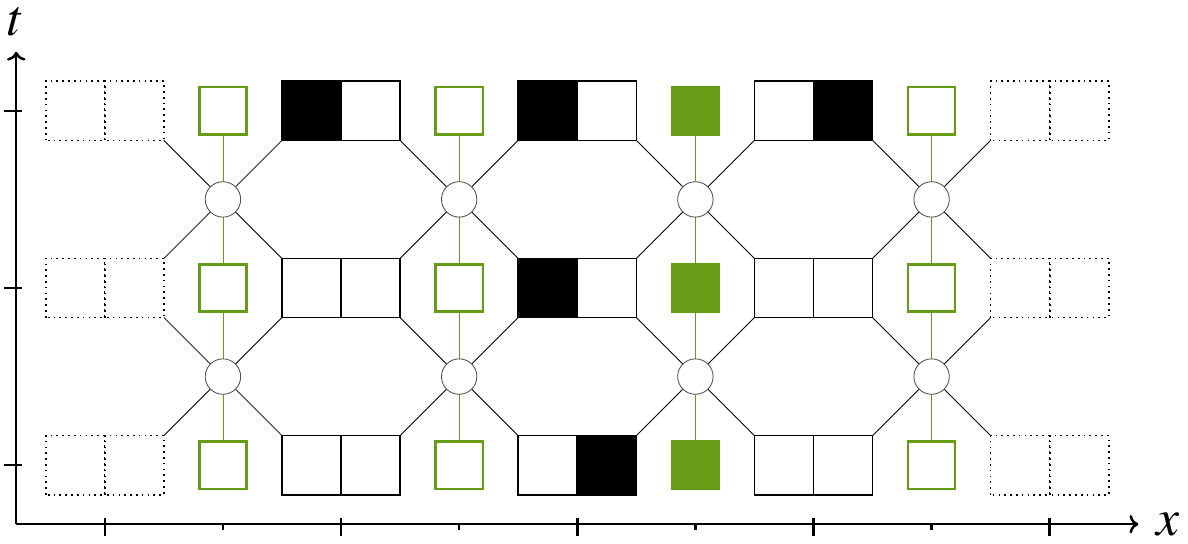}}
    \caption{\label{fig:rcaa}Three space-time diagrams using the running
        example (1) of Sec-\ref{sec:proc}.}
\end{figure}
\end{example*}

\begin{example*}[2]
In Physics, a distinction is often made between abelian and non-abelian gauge 
theories. The abelian gauge theory accounts for quantum electrodynamics (QED)
while non-abelian gauge theory accounts, for example,
for quantum chromodynamics (QCD).

It is easy to extend example (1) to become non-abelian by extending the
alphabet $\Sigma$ to $\{1,...,N\}^2$ for $N>2$. $G$ (the 
local gauge transformation restricted to the states $\Sigma$) is again a set
of permutations that change both components of the state in the same way. In 
fact we take all of them:
$$G=\{\sigma \otimes \sigma \; | \; \sigma \in S(N)\}$$
with $S(N)$ the set of permutations over $N$ elements. We can then define
the gauge field to again be $\Lambda=G$ ; the extended local gauge 
transformations $\overline{G}$ through Eq. \eqref{eq:gaugetransfall}
and the set of gauge transformations $\Gamma$ through 
definition-\ref{def:gaugetransf}.
Two local gauge transformations $g_x$ and $g_{x+1}$ do commute for the 
same reason as in the abelian case, thus $\Gamma$ is well defined.

Using those definitions, the local rule defined in Eq. \eqref{eq:rfinal} is
already gauge-invariant. This can be checked in a straightforward manner
through the exact same procedure as the abelian case.

Fig. \ref{fig:rcana} gives two space-time diagrams of a
non-abelian gauge-invariant CA with $N=3$---i.e. $|\Sigma|=9$ and $|G|=3!=6$.
An empty state for the gauge field 
represents the identity and a filled state represents $\tau\otimes\tau$
where $\tau$ is the permutation between black and white, leaving the 
gray untouched. Only those two state are represented here to keep the 
figure readable, even though there are more transformations available.

Sub-figure \ref{sub@subfig:rcana1} 
features an example of two "particles" crossing. Sub-figure 
\ref{sub@subfig:rcana2} represents just the same scenario, only with a gauge 
transformation $\tau\otimes\tau$ has been applied in the middle.

\begin{figure}[ht!]
    \subfloat[Two particles cross in an empty gauge field
        \label{subfig:rcana1}]
    {\includegraphics[width=0.49\textwidth]{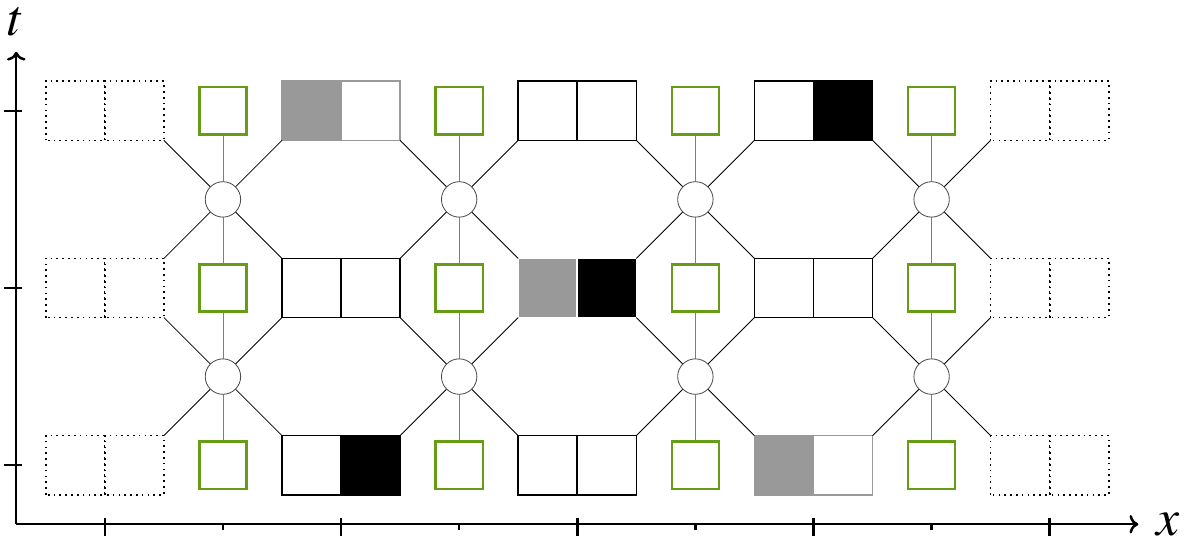}}
    \hfill
    \subfloat[Equivalent up to a gauge transformation
        \label{subfig:rcana2}]
    {\includegraphics[width=0.49\textwidth]{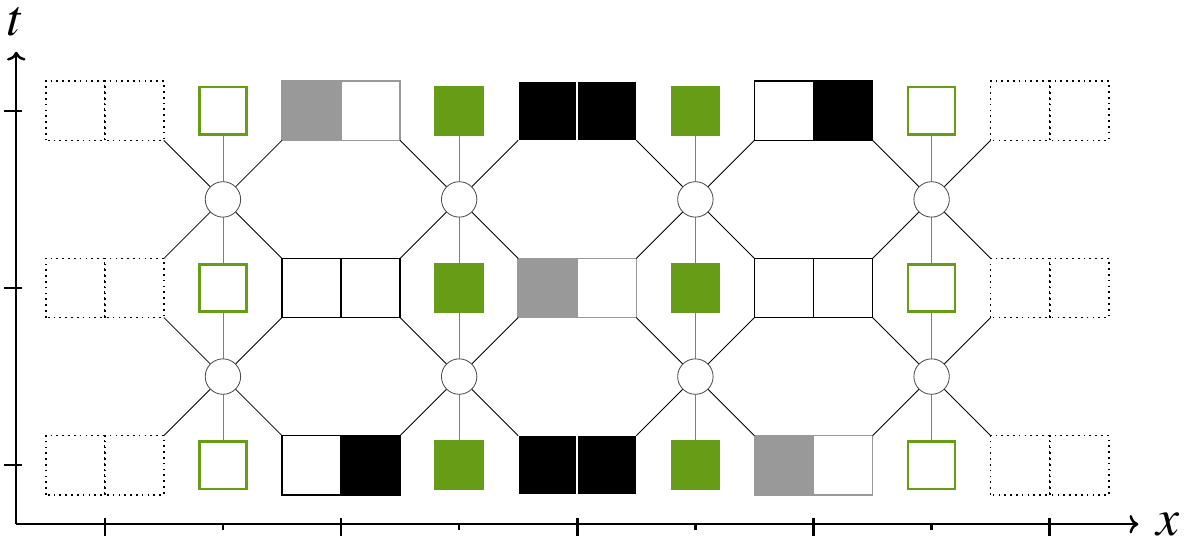}}
    \newline \centering
    \subfloat[New behaviours emerge from having the gauge field
        \label{subfig:rcana3}]
    {\includegraphics[width=.49\textwidth]{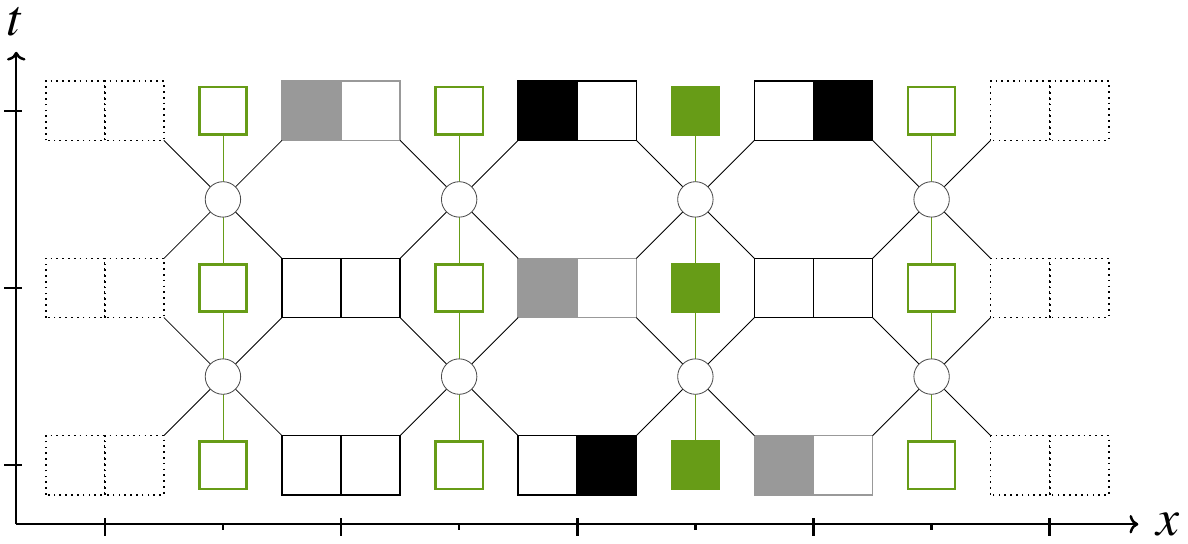}}
    \caption{\label{fig:rcana}Three examples of dynamics for
        a non-abelian version of the running example.}
\end{figure}
\end{example*}

\begin{example*}[3]
Now switching from the classical to the quantum setting, the same gauging
procedure has been applied to obtain a gauge-invariant QCA
\cite{Arrighi2020}, in the abelian case. This example will not be treated in 
detail, only the definition of the model and the reason for it being 
gauge-invariant will be given here. 
 
The classical configuration are obtained again from alphabet $\Sigma=\{0,1\}^2$. 
These two boolean numbers code the presence, or the absence of two 
fermions. The pseudo-spin of the fermion is encoded by the choice of the 
component, i.e. $01$ is the left-moving spin, $10$ is the right-moving spin. 
There can be two fermions of opposite spin, which is the state $11$, but there 
cannot be two fermions of the same spin, by the Pauli exclusion principle. 
This is sometimes referred to as the occupation number representation.
The gauge field can be seen as a counter of particles.

\paragraph{Evolution rule.} $r$ is now a unitary matrix that acts on 
$\mathcal{H}_\Sigma \otimes \mathcal{H}_\Lambda \otimes \mathcal{H}_\Sigma$
\begin{equation}
    r = \begin{pmatrix}
        I & 0 & 0 & 0\\
        0 & -is I & cV & 0 \\
        0 & cV^\dagger & -isI & 0 \\
        0 & 0 & 0 & -I
    \end{pmatrix} e^{\frac{i}{2}\epsilon^2g^2L^2}
\end{equation}
with $g$ a parameter called the charge, 
$\epsilon$ the distance between two points in space,
$I$ the identity,
$s$ and $c$ stands for $sin(m\epsilon)$ and $cos(m\epsilon$)
respectively where $m$ is the mass,
and for $l\in\mathbb{Z}$:
\begin{itemize}
    \item[ ] $V\ket{l} = \ket{l-1}$
    \item[ ] $L\ket{l} = l\ket{l}$.
\end{itemize} 

This evolution rule defines the dynamics directly for both the fermions and 
the gauge field. When $g=0$ this dynamics is just a non-interacting, 
multi-particle version of the Dirac quantum walk. If, furthermore,
the mass is zero, the fermions do not change direction.
The operator $V$ and its conjugate transpose $V^\dagger$ allows for 
the gauge field to act as a counter
of fermions that go through the link between two nodes. 
The exponential term is given here so as to 
be complete. It is the one which creates the interaction between the gauge 
field and the fermions. It will have no impact in the proof of 
gauge-invariance.

The minus sign of the bottom-right entry of the matrix is needed in
the qubit representation of two fermions crossing past each other. 
For the same reason, right before each new cell is formed, 
the following gate is applied on the incoming components.
\begin{equation}
    S=\begin{pmatrix}
        \;1\; & \;0\; & \;0\; & \;0\; \\
        0 & 1 & 0 & 0 \\
        0 & 0 & 1 & 0 \\
        0 & 0 & 0 & -1 \\
    \end{pmatrix}
\end{equation}

\paragraph{Gauge transformations.}
For $x$ a position, the local gauge transformation $g_x$ acts on the 
fermionic field at position $x$ and on the gauge field in positions  $(x-1,x)$
and $(x,x+1)$ which is reminiscent of the classical case. 
For $\varphi\in\mathbb{R}$, we define 
\begin{align}
    R_\varphi : &\ket{0} \mapsto \ket{0}\\
                &\ket{1} \mapsto e^{i\varphi}\ket{1}
\end{align}
which is a $U(1)$ gauge transformation, and for $l\in \mathbb{Z}$:
\begin{equation}
    T_\varphi \ket{l} = e^{i l\varphi} \ket{l}.
\end{equation}

Then the group $G$ is defined using $R_\varphi$ and $T_\varphi$ for any 
$\varphi$ in $\mathbb{R}$, where $R_\varphi$ is applied on each component of a
state in $\Sigma$ and $T_\varphi$ is applied to the gauge field:
\begin{equation}
    G = \Big\{T_\varphi \otimes \big(R_\varphi \otimes R_\varphi\big)
        \otimes T_{-\varphi} \; | \; \varphi \in \mathbb{R} \Big\}
\end{equation}

The local gauge transformation at positions $x$ and $y$ commute, 
so the gauge transformation set $\Gamma$ can again be defined through 
definition-\ref{def:gaugetransf} as the set of operators that apply, locally,
any local gauge transformation.

\paragraph{Gauge-invariance.}
This model is gauge-invariant with $Z$ as the identity. In order to prove it 
let us focus on the input $\ket{mln}$ of a gate. When applying a gauge 
transformation $g_\varphi$, this input state will trigger a phase gain 
$\mathcal{V}(x,m-l,n+l)$:
\begin{align}
    &m\varphi(x) + l\big(\varphi(x+1) - \varphi(x)\big) + n \varphi(x+1)\\
        &=  (m-l)\varphi(x) + (n+l)\varphi(x+1).
\end{align}
However, the numbers $(m-l, n+l)$ are invariants of $r$, as it takes 
$\ket{mln}$ into a superposition of the form:
\begin{equation}
    \sum_{i\in\{-1,0,1\}}\alpha_i \ket{m-i, l-i, n+i}.
\end{equation}
It follows that the phase gain will be the same whether the gauge 
transformation is applied before or after the 
evolution and thus, the evolution commutes with the gauge transformations. 
Hence, the QCA is gauge-invariant under $\gamma$ and for $Z$ being
the identity.

\paragraph{Physical model.}
This QCA is quite specific because it was conceived so that it provides 
a discrete space-time formulation of one-dimensional quantum electrodynamics, 
which is one of the four fundamental interactions in Physics 
\cite{Arrighi2020}.
The continuous limit for this specific model has not been formally derived, 
however this limit has been done in the free case both for continuous time
 \cite{di2020quantum,Manighalam2021ContinuousTL} and the same method could potentially work 
in this case.
 

\end{example*}

\section{Degrees of freedom induced by gauge-invariance}
\label{sec:theorems}

\subsection{Equivalence of theories}

Given a set of gauge transformations $\Gamma$, multiple CA may
lead to equivalent dynamics up to $\Gamma$: 
\begin{definition}[Equivalence of gauge-invariant CA] \label{def:equiv}
    Let $T$ be a gauge-invariant CA with respect to a given $\Gamma$
    and $Z$.
    $T$ is simulated by a CA $T'$ if and only if for each element $c$ of 
    $\mathcal{C}$ (or $\mathcal{H}$ in the quantum case)
    there exists $\gamma, \gamma' \in \Gamma$ such that $(\gamma \circ T)(c)
    = (T'\circ \gamma')(c)$. They are equivalent if both simulate each other. 
    Equivalence will be denoted $T\equiv T'$. 
\end{definition}

In practice, $T$ is gauge-invariant with respect to a specific $\Gamma$ and 
$Z$. Adding a constraint on $Z$, one may characterize the equivalence
of two CA using different quantifiers and constraints which 
may be useful for some specific problems.

\begin{proposition}[Characterization of equivalence for gauge-invariant CA]
Let $T$ be a gauge-invariant CA with respect to $\Gamma$ and $Z$ and $T'$
another CA over the same alphabet as $T$.
If $Z$ is reversible and $T'$ is gauge-invariant with 
respect to $\Gamma$ and $Z$, then these three statements
are equivalent:
\begin{enumerate}
    \item $T$ is simulated by $T'$.
    \item $\forall c, \exists \gamma \in \Gamma$ such 
        that $T(c) = T'\circ \gamma (c)$.
    \item $\forall c, \forall \gamma \in \Gamma$, 
        $\exists \gamma' \in \Gamma$ such that $\gamma
        \circ T(c) = T'\circ \gamma' (c)$.
\end{enumerate}
\end{proposition}

\begin{proof}
    We shall prove the equivalence through three implications. The proof is 
    given in the classical setting, but carries through to the quantum case 
    where $c$ takes its value 
    in $\mathcal{H}$ instead of $\mathcal{C}$.
    \begin{itemize}
        \item Suppose (1), then for $c$ a configuration, 
            we have $\gamma, \gamma' \in \Gamma$ such that 
            $(\gamma \circ T)(c) = (T'\circ \gamma')(c)$. 
            But since $\Gamma$ is a group, it implies that
            $T(c) = (\gamma^{-1} \circ T' \circ \gamma')(c)$.
            And since $Z$ is reversible, we obtain $T(c) = 
            (T'\circ Z^{-1} (\gamma^{-1}) \circ \gamma') (c)$.
            However, $Z^{-1} (\gamma^{-1}) \circ \gamma'$ is an
            element of $\Gamma$ therefore we have proven that (1)
            implies (2).
        \item Suppose (2), let $c$ be a configuration and take
            $\gamma\in\Gamma$ such that $T(c)=(T'\circ\gamma)(c)$. 
            Since $\Gamma$ is a group, for any $\gamma_1 \in \Gamma$
            there exists $\gamma_3 \in \Gamma$ such that 
            $\gamma = \gamma_3 \circ \gamma_1$.
            Therefore, from gauge-invariance of $T'$, 
            $T(c) = (Z(\gamma_3) \circ T' \circ \gamma_1) (c)$
            which is equivalent to $(Z(\gamma_3)^{-1} \circ T)(c) = 
            (T\circ \gamma_1) (c)$ because $G$ is a group. 
            And writing $\gamma_2 = Z(\gamma_3)^{-1}$
            which is in $\Gamma$, we conclude that (2) implies (3).
        \item The fact that (3) implies (1) is immediate because (3) is 
            a generalization of (1): both statements differ only 
            by the quantifier before $\gamma$. If for any $\gamma$ the 
            property is true, then it is also true for one specific $\gamma$.
    \end{itemize} \qed
\end{proof}
    
\begin{example*}[1]
Fig. \ref{fig:equiv} shows two diagrams of equivalent CA. The local 
rule $r$ driving the evolution in sub-figure \ref{sub@subfig:equiv1} is the 
one given in Eq.\eqref{eq:rfinal} for the abelian example
whilst the local rule driving the evolution
of sub-figure \ref{sub@subfig:equiv2} is $r'=\gamma \circ r$ where 
$\gamma$ is the gauge transformation that swaps black and 
white everywhere (note that this does not impact the gauge field because the 
transformation on each side of it cancel out).

\begin{figure}[ht!]
    \subfloat[Dynamics developed previously\label{subfig:equiv1}]
    {\includegraphics[width=0.49\textwidth]{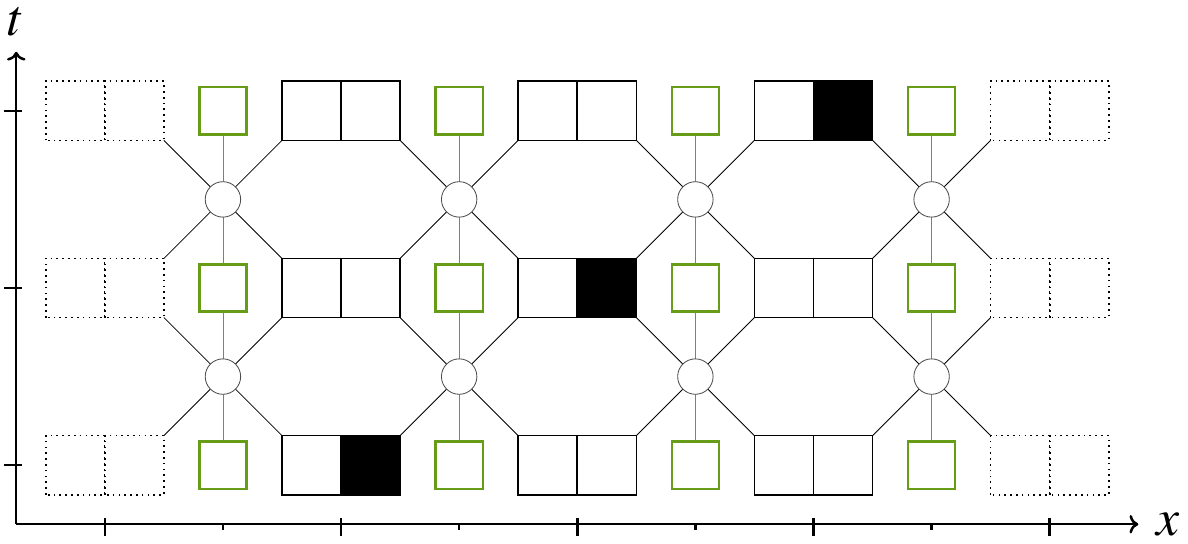}}
    \hfill
    \subfloat[A gauge-equivalent dynamics
        \label{subfig:equiv2}]
    {\includegraphics[width=0.49\textwidth]{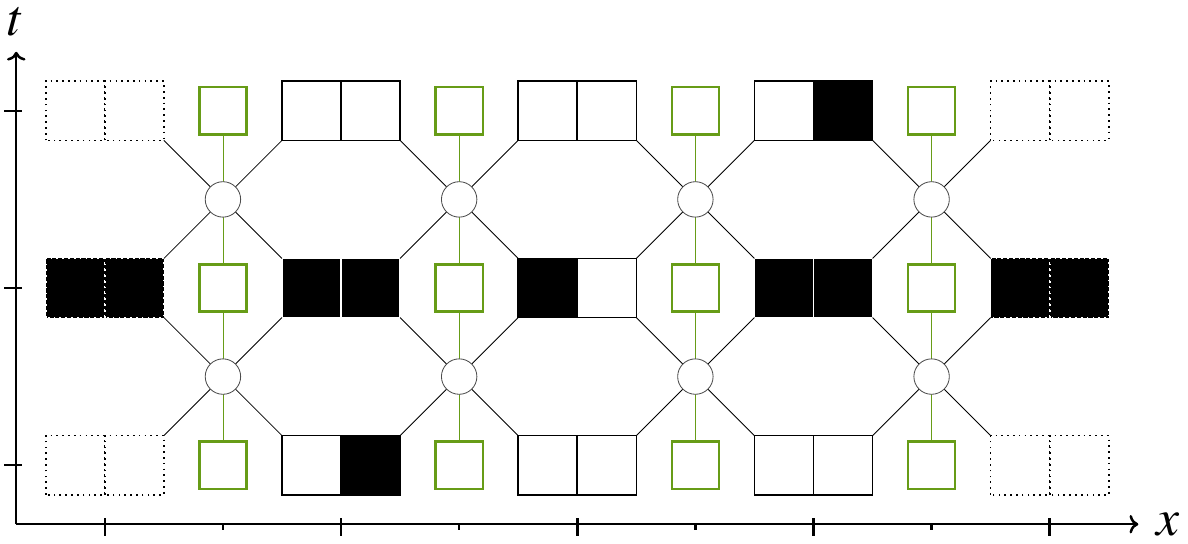}}

    \caption{\label{fig:equiv}Example of two equivalent CA.}
\end{figure}
\end{example*}

\subsection{Gauge-fixing and gauge-constraining}
Gauge-invariance states that there is a degree of freedom in both the 
dynamics---equivalence of CA---and the set of possible states. 
Gauge-fixing means choosing, amongst the many possible dynamics. 
Gauge-constraining, on the other hand, 
deals with the removal of the degrees of freedom, e.g. via the choice of a 
canonical representative element for each class of gauge 
equivalent configurations.

\paragraph{Gauge-fixing.}
With equivalence of CA, one can use many different representations 
for the same evolution model: two equivalent CA will model the same dynamics 
up to a gauge transformation. Therefore, there is a degree of freedom
in choosing a specific CA as a model for a specific dynamics. 
Choosing this degree of freedom is called gauge-fixing.

In other words, the explicit evolution scheme is undetermined because of 
the gauge-invariance: if a configuration $c$ at times $t$ evolves into
a configuration $c'$, it is the same as if it evolved into $\gamma(c')$
for $\gamma$ a gauge transformation. Gauge-fixing is the choice of an 
explicit evolution scheme. 


\paragraph{Gauge-constraining.}
The redundancy induced by gauge-invariance is somewhat problematic, because
it is there in the state space, when really it should not be observable.
In other words, a gauge transformation 
should not alter the observation. In quantum mechanics, for a system described 
by a density matrix $\rho$, what is being measured for $\mathcal{O}$ an 
observable, is the expected value: 
$$Tr(\mathcal{O}\rho).$$
A gauge transformation $\gamma$ is a unitary that will act on a density matrix
as follows :
$$ \rho \mapsto \gamma \rho \gamma^\dagger$$
with $\gamma^\dagger$ being the conjugate transpose of $\gamma$.

One may wish to restrict the
states, or observables, to being physical. That is to say 
for any gauge transformation $\gamma$, the equality $Tr(\mathcal{O} \,\rho) =
Tr\big(\mathcal{O} \gamma \rho \gamma^\dagger\big)$ should hold. 
There are two ways to do so,
one is to restrict the set of allowed observables to those which commute with 
the gauge transformations. Indeed, by the cyclic property of the trace, the 
action of $\gamma$ and $\gamma^\dagger$ will then cancel out.

Another way is to allow any set of observables but to restrict the states 
to those which commute with gauge transformations---i.e. $[\gamma,
\rho] = 0$ for any gauge transformation $\gamma$. This is called 
gauge-constraining. Again, this would lead to $\gamma$ and $\gamma^\dagger$ 
cancelling out.

Gauge-constraining is frequently used in Physics, but still remains 
to be formulated for classical CA.

\section{Conclusion}
\label{sec:conclu}
\paragraph{Summary.}
The paper followed a constructive approach to gauge-invariance in Cellular 
Automata (CA). Formal definitions of CA, gauge transformations and 
gauge-invariance were given. For $R$ a CA and $\Gamma$ a set of gauge 
transformations, it was shown how to obtain a gauge-invariant CA $T$ through 
the introduction of a gauge-field, whilst keeping $R$ as a sub-case. 
The extension of $R$ into $T$ is called 
gauging procedure and comes from Physics. Three examples of gauge-invariant 
CA were then given, for abelian and non-abelian 
gauge transformation in classical CA, and for abelian gauge transformations 
in a QCA. Because of the redundancy inherent to 
gauge-invariance, CA implementing a gauge-invariance with respect to the same
set of gauge transformations can have the same dynamics up to gauge 
transformations, which means those two theories are equivalent. The equivalence
of CA was defined and characterized. Finally, 
gauge-fixing and gauge-constraining were introduced as ways 
for choosing or removing the degrees of freedom induced by gauge-invariance.

\paragraph{Related works.}
A number of discrete counterparts to Physics symmetries have been 
reformulated in terms of CA, including reversibility, 
Lorentz-covariance \cite{arrighi2014discrete} and conservations laws and 
invariants \cite{formenti2011hierarchy}. To our 
knowledge the closest work is the colour-blind CA construction \cite{salo2013} 
which implements a global colour symmetry without porting it to the local 
scale. However gauge symmetries have been implemented in the one-particle 
sector of Quantum CA, a.k.a for Quantum Walks
\cite{arnaultdebbasch2016quantum,arnault2017discrete,cedzich2019quantum}. One of the authors 
had followed a similar procedure in order to introduce the electromagnetic 
gauge field \cite{di2014quantum,marquez2018electromagnetic}, and that of the weak and strong interactions 
\cite{arnault2016quantum,di2016quantum}. This again was done in the very fabric 
of the Quantum Walk and the associated symmetry was therefore an intrinsic 
property of the Quantum Walk. But the gauge field would remain continuous, 
and seen as an external field. Recently, this symmetry has been studied in the
classical CA for both abelian and non-abelian gauge-invariance 
\cite{arrighi2018gauge,arrighi2019non} and for a QCA as well \cite{Arrighi2020}.

There are, of course, numerous other approaches to space-discretized gauge theories, the main ones being Lattice Gauge Theory 
\cite{banuls2017efficient,emonts2020gauss,kaplan2020gauss,rothe2012lattice} and the Quantum Link Model \cite{chandrasekharan1997quantum}, which were phrased in terms of Quantum Computation--friendly terms through Tensor 
Networks \cite{rico2014tensor,zohar2018combining} and can be linked in a unified framework \cite{silvi2014}. The tensor network representation is based on the locality of the Hamiltonian evolution, the system is not described in terms of large vectors, but with local tensors and boundary constraints, allowing for a more efficient description. Changing the parameters of the tensor using variational optimization techniques---e.g. gradient descent---over the expected energy $\bra{\psi} H \ket{\psi}$ allows to recover the ground state which describes the low-energy Physics of the system \cite{ercolessi2018phase,felser2020two,magnifico2020real,magnifico2021lattice,notarnicola2015discrete}. Discretized gauge-theories have also arisen from other horizons, such as Ising models \cite{silvi2014,wegner1971}. With the recent progress in quantum hardware---e.g. trapped ions, Rydberg atoms, quantum superconducting---quantum simulations can also be directly implemented on quantum technologies \cite{banuls2020simulating,klco20202}.

All of these approaches, however, begin with a well-known continuous gauge 
theory which is then space-discretized---time is usually kept continuous. 
An interesting attempt to quantum discretize gauge theories in discrete time,
on a general simplicial complex can be found in \cite{kornyak2009discrete}.

At the time of the writing, some theoretical questions about 
gauge-invariant cellular automata remained open. 
Is it always possible to make a CA gauge-invariant under any 
groups of gauge transformation? How do this model relate to 
color-blind cellular automata \cite{salo2013}?
These two questions have since then been answered by two of 
the authors \cite{arrighi_et_al:LIPIcs.MFCS.2021.9}.

\paragraph{Perspectives.}  The hereby developed methodology 
has already been applied to Quantum CA (QCA) \cite{Arrighi2020} with abelian
gauge symmetry,
so as to obtain the Schwinger model for quantum electrodynamics. We 
believe it can be further extended to non-abelian gauge-invariance in 
order to have, for example, a discrete counterpart to quantum chromodynamics,
and to 2 or 3 dimensions in space so as to expand the possible dynamics.
Such discretized theories may be of interest in 
Physics especially in non-perturbative theories
\cite{strocchi2013introduction}, but they may also represent practical assets 
as quantum simulation algorithms, i.e. numerical schemes that run on Quantum 
Computers to efficiently simulate interacting fundamental particles 
theories---a task which would take a very long time on classical computers.

Another perspective lies in error-correction. Taking the set of 
gauge transformations to be the set of errors, we expect that gauge-invariance
provides a way to obtain a dynamics that would be error-free. Applied in the 
context of CA, this would correspond to error-correction in spatially 
distributed systems.

\begin{acknowledgements}
    NE would like to thank Simone Montangero and Giuseppe Magnifico for
    helpful discussions. 
    
    GDM acknowledges the invitational Fellowships for
    Research in Japan supported by Japan Society for the Promotion of 
    Science (JSPS), Bridge Program (ID: BR190201).

    This publication was made possible through the support of the ID\# 61466
    grant from the John Templeton Foundation, as part of the “The Quantum 
    Information Structure of Spacetime (QISS)” Project (qiss.fr). The opinions 
    expressed in this publication are those of the author(s) and do not 
    necessarily reflect the views of the John Templeton Foundation.
\end{acknowledgements}

%
%

\bibliographystyle{spmpsci}      
\bibliography{biblio}   

\end{document}